\documentclass[pra,aps,twocolumn,showpacs,preprintnumbers]{revtex4}
\usepackage{graphics,amsmath,amsfonts,amscd,revsymb,latexsym, enumerate,multirow,epsfig}

\newtheorem{theorem}{Theorem}

\newtheorem{corollary}[theorem]{Corollary}

\newtheorem{proposition}[theorem]{Proposition}

\newcommand{\qed}{{\hfill$\Box$}}
\newenvironment{proof}{\noindent \textbf{{Proof~} }}{\qed}
\def\bi{\begin{itemize}}
\def\ei{\end{itemize}}
\def\be{\begin{equation}}
\def\ee{\end{equation}}
\def\bea{\begin{eqnarray}}
\def\eea{\end{eqnarray}}
\def\ben{\begin{eqnarray*}}
\def\een{\end{eqnarray*}}

\def\>{\rangle}
\def\<{\langle}

\newcommand{\1} I

\def\*{\star}

\def\cS{{\cal S}}

\def\0{{\mathbf{0}}}
\def\1{{\mathbf{1}}}
\def\2{{\mathbf{2}}}
\def\3{{\mathbf{3}}}
\def\4{{\mathbf{4}}}
\def\5{{\mathbf{5}}}
\def\6{{\mathbf{6}}}
\def\7{{\mathbf{7}}}
\def\8{{\mathbf{8}}}
\def\9{{\mathbf{9}}}

\begin{document}

\title{Classical Enhancement of Quantum Error-Correcting Codes}
\author{Isaac Kremsky}
\affiliation{Physics Department, University of Southern California, Los Angeles, CA 90089}
\author{Min-Hsiu Hsieh}
\thanks{To whom correspondence should be addressed}
\email{minhsiuh@gmail.com}
\author{Todd A. Brun}
\affiliation{Ming Hsieh Electrical Engineering Department, University of Southern
California, Los Angeles, CA 90089}
\date{\today}

\begin{abstract}
We present a general formalism for quantum error-correcting codes that
encode both classical and quantum information (the EACQ formalism). This
formalism unifies the entanglement-assisted formalism and classical error
correction, and includes encoding, error correction, and decoding steps such
that the encoded quantum and classical information can be correctly
recovered by the receiver. We formally define this kind of quantum code
using both stabilizer and symplectic language, and derive the appropriate
error-correcting conditions. We give several examples to demonstrate the
construction of such codes.
\end{abstract}

\maketitle

\section{Introduction}

\label{sec1}

Since Shor proposed the first quantum error correction code (QECC) \cite%
{Shor95}, research in this field has progressed rapidly. A broad theory of
quantum error-correcting codes was created with the stabilizer formalism and
its symplectic formulation \cite{CRSS98,DG97thesis}, that allow the
systematic description of a large class of quantum error correction codes
and their error-correcting properties. In this formulation, a QECC is
defined to be a subspace fixed by a \textit{stabilizer group}. At the same
time, a construction of QECCs from classical error correction codes was
proposed separately by Calderbank, Shor and Steane \cite{CS96,Ste96}, the
so-called CSS construction. Later this was generalized to give a stronger
connection between quantum codes and classical symplectic codes; however, it
seemed that this connection between quantum coding theory and classical
coding theory was not universal, since only certain symplectic codes
possessed quantum equivalents.

More recent developments in quantum coding theory have led to the
development of the operator quantum error correction formalism (OQECC) \cite%
{AKS06a,Bacon05a,KS06a,DRD05,Kribs06,NP05,DB05} and the
entanglement-assisted quantum error correction formalism (EAQECC) \cite%
{Bow02,EAQECC2,BDH06}; moreover, it is possible to produce a unified
formalism (EAOQECC) \cite{HBD07} that combines both OQECCs and EAQECCs. This
formalism demonstrates that a broader connection exists between classical
and quantum error correction theory. Good QECCs can be obtained by a
generalized CSS construction from good classical codes. This opens the door,
for example, to the construction of high-quality quantum codes from modern
classical codes, such as Turbo and LDPC codes \cite{HBDldpc07}.

In this paper, we generalize this construction in a different way, by
proposing new quantum codes that can be used to transmit both classical and
quantum information simultaneously. We call this scheme the
entanglement-assisted, classically enhanced quantum error correction
formalism, but throughout the paper it will be referred to simply as the
EACQ formalism . The EACQ formalism can be considered a generalization
EAQECCs, or as a unification of quantum and classical linear error
correction codes. This unification also makes contact with results in
quantum information theory, where bounds exist on the asymptotic
transmission of simultaneous classical and quantum information, including
the use of entanglement assistance. It is believed that these bounds are
better than simple time-sharing between codes for transmitting quantum and
classical information separately through a quantum channel \cite{prep2008dev}%
. It is our hope that it may be possible to construct classes of codes which
achieve these rates in the limit of large block size.

This paper is organized as follows. We give a brief introduction of the
EAQECC formalism using both the stabilizer and the symplectic language in
section \ref{sec2}. In section \ref{sec3}, we formally define a quantum code
(EACQ) that can transmit both classical and quantum information at the same
time. Several properties of this kind of quantum code are also discussed in
this section. We provide several examples in section \ref{sec4}, to
demonstrate the usefulness of this formalism. We conclude in section \ref%
{sec5} by examining some special cases, and arguing that the EACQ formalism
is indeed a generalization and unification of quantum and classical coding
theory.

\section{EAQECC}

\label{sec2}

In this section, we will review entanglement-assisted quantum error
correction using both stabilizer and symplectic language.

Let $\mathcal{G}_{n}$ be the $n$-fold Pauli Group \cite{NC00}. Every
operator in $\mathcal{G}_{n}$ has either eigenvalues $\pm 1$ or $\pm i$. An $%
[[n,q,d;e]]$ EAQECC is a quantum code that encodes $q$ logical quantum bits
(qubits) into $n$ physical qubits with the help of $e$ maximally entangled
pairs (ebits) shared between sender and receiver, and can correct up to $%
\lfloor d/2\rfloor $ single-qubit errors. Such an EAQECC is defined by a
non-commuting group $\mathcal{S}_{Q}=\langle \overline{Z}_{1},\cdots ,%
\overline{Z}_{s},\overline{Z}_{s+1},\overline{X}_{s+1},\cdots ,\overline{Z}%
_{s+e},\overline{X}_{s+e}\rangle \subset \mathcal{G}_{n}$ of size $2^{s+2e}$%
, where $s+e+q=n$, and the generators $\overline{Z}_{i}$ and $\overline{X}%
_{i}$ satisfy the following commutation relations:
\begin{equation}
\begin{split}
\lbrack \overline{Z}_{i},\overline{Z}_{j}]& =0\ \ \ \ \forall i,j \\
\lbrack \overline{X}_{i},\overline{X}_{j}]& =0\ \ \ \ \forall i,j \\
\lbrack \overline{X}_{i},\overline{Z}_{j}]& =0\ \ \ \ \forall i\neq j \\
\{\overline{X}_{i},\overline{Z}_{i}\}& =0\ \ \ \ \forall i.
\end{split}
\label{comm}
\end{equation}%
We define the isotropic subgroup $\mathcal{S}_{Q,I}$ of $\mathcal{S}_{Q}$ to
be the subgroup generated by $\{\overline{Z}_{1},\cdots ,\overline{Z}_{s}\}$%
; it is of size $2^{s}$. Similarly, the symplectic subgroup $\mathcal{S}%
_{Q,S}$ of $\mathcal{S}_{Q}$ is of size $2^{2e}$ and is generated by $\{%
\overline{Z}_{s+1},\overline{X}_{s+1},\cdots ,\overline{Z}_{s+e},\overline{X}%
_{s+e}\}$. The isotropic subgroup $\mathcal{S}_{Q,I}$ is Abelian; however,
the symplectic subgroup $\mathcal{S}_{Q,S}$ is not. We can easily construct
an Abelian extension of $\mathcal{S}_{Q,S}$ that acts on $n+e$ qubits, by
specifying the following generators:
\begin{eqnarray*}
\overline{Z}_{1} &\otimes &I, \\
&\vdots & \\
\overline{Z}_{s} &\otimes &I, \\
\overline{Z}_{s+1} &\otimes &Z_{1}, \\
\overline{X}_{s+1} &\otimes &X_{1}. \\
&\vdots & \\
\overline{Z}_{s+e} &\otimes &Z_{e}, \\
\overline{X}_{s+e} &\otimes &X_{e},
\end{eqnarray*}%
where the first $n$ qubits are on the side of the sender (Alice) and the
extra $e$ qubits are taken to be on the side of the receiver (Bob). The
operators $Z_{i}$ or $X_{i}$ to the right of the tensor product symbol above
is the Pauli operator $Z$ or $X$ acting on Bob's $i$-th qubit. The picture
is that Alice and Bob initially share $e$ ebits; Alice encodes her $q$
qubits together with her halves of the $e$ entangled pairs and $s$ ancilla
qubits. Bob's qubits are his halves of the $e$ entangled pairs. Because this
code assumes pre-existing entanglement between Alice and Bob, it is an
entanglement-assisted quantum error-correcting code (EAQECC). We denote such
an Abelian extension of the group $\mathcal{S}_{Q,S}$ by $\widetilde{%
\mathcal{S}}_{Q,S}$. This EAQECC can correct an error set ${\mathbf{E}}%
\subset \mathcal{G}_{n}$ if for all $E_{1},E_{2}\in {\mathbf{E}}$, $%
E_{2}^{\dagger }E_{1}\in \mathcal{S}_{Q,I}\cup (\mathcal{G}_{n}-N(\mathcal{S}%
_{Q}))$, where $N(\mathcal{S})$ is the normalizer of group $\mathcal{S}$.

Before we describe EAQECCs using the symplectic language, we need to first
discuss some of the basic properties of the symplectic form which are
relevant to the discussion that follows. The symplectic form of vectors in $(%
\mathbb{Z}_{2})^{2n}$ is useful for specifying Pauli operators on $n$ qubits
when the global phase may be ignored. We write a vector ${\mathbf{u}}\in (%
\mathbb{Z}_{2})^{2n}$ in symplectic form by splitting it into two vectors ${%
\mathbf{x}},{\mathbf{z}}\in (\mathbb{Z}_{2})^{n}$ and writing it as follows:
${\mathbf{u}}=({\mathbf{z}}|{\mathbf{x}})$. We define
\begin{equation*}
N_{({\mathbf{z}}|{\mathbf{x}})}\equiv Z^{z_{1}}X^{x_{1}}\otimes
Z^{z_{2}}X^{x_{2}}\otimes \cdots \otimes Z^{z_{n}}X^{x_{n}},
\end{equation*}%
where $z_{r}$ $(x_{r})$ is the $r$-th element of ${\mathbf{z}}$ $({\mathbf{x}%
})$. Thus a set of $m$ Pauli-operators acting on $n$ qubits may be specified
by a matrix with $m$ rows ${\mathbf{u}}_{i}\in (\mathbb{Z}%
_{2})^{2n},i=1,2,\cdots ,m$. The symplectic product between two vectors is
defined as
\begin{equation*}
({\mathbf{z}}|{\mathbf{x}})\odot ({\mathbf{z}}^{\prime }|{\mathbf{x}}%
^{\prime })={\mathbf{z}}\cdot {\mathbf{x}}^{\prime T}-{\mathbf{x}}\cdot {%
\mathbf{z}}^{\prime T}.
\end{equation*}%
(Note that in the binary case, as here, subtraction is the same as
addition.) Two Pauli operators $N_{({\mathbf{z}}|{\mathbf{x}})}$ and $N_{({%
\mathbf{z}}^{\prime }|{\mathbf{x}}^{\prime })}$ commute if and only if $({%
\mathbf{z}}|{\mathbf{x}})\odot ({\mathbf{z}}^{\prime }|{\mathbf{x}}^{\prime
})=0$.

Recall that the stabilizer $\mathcal{S}_{Q}$ of an $[[n,q;e]]$ EAQECC is
generated by $s+2e$ elements. Therefore, it can be specified by an $%
(s+2e)\times 2n$ symplectic matrix, $\hat{F}$, which we will refer to as the
\textit{quantum parity check matrix} in this paper. Thus,
\begin{equation}
\mathcal{S}_{Q}=\{N_{{\mathbf{u}}}|{\mathbf{u}}\in \text{Rowspace}(\hat{F}%
)\},  \label{Sqstabilizer}
\end{equation}%
where
\begin{equation}
\hat{F}=\left(
\begin{array}{c}
{\mathbf{u}}_{1} \\
\vdots \\
{\mathbf{u}}_{s+e} \\
{\mathbf{v}}_{s+1} \\
\vdots \\
{\mathbf{v}}_{s+e}%
\end{array}%
\right) .  \label{F_hat}
\end{equation}%
In this matrix, the rows ${\mathbf{u}}_{1}\cdots {\mathbf{u}}_{s+e}$
represent the generators $\overline{Z}_{1}\cdots \overline{Z}_{s+e}$, and
the rows ${\mathbf{v}}_{s+1}\cdots {\mathbf{v}}_{s+e}$ represent $\overline{X%
}_{s+1}\cdots \overline{X}_{s+e}$. The commutation relations in (\ref{comm})
translate to the following:
\begin{equation}
\begin{split}
{\mathbf{u}}_{i}\odot {\mathbf{u}}_{j}& =0\ \ \ \ \forall i,j \\
{\mathbf{v}}_{i}\odot {\mathbf{v}}_{j}& =0\ \ \ \ \forall i,j \\
{\mathbf{u}}_{i}\odot {\mathbf{v}}_{j}& =0\ \ \ \ \forall i\neq j \\
{\mathbf{u}}_{i}\odot {\mathbf{v}}_{i}& =1\ \ \ \ \forall i.
\end{split}
\label{comm_sp}
\end{equation}%
The isotropic subgroup $\mathcal{S}_{Q,I}$ and the symplectic subgroup $%
\mathcal{S}_{Q,S}$ can be rewritten as:
\begin{equation*}
\begin{split}
\mathcal{S}_{Q,I}& =\{N_{{\mathbf{u}}}|{\mathbf{u}}\in \text{Rowspace}(\hat{F%
}_{I})\}, \\
\mathcal{S}_{Q,S}& =\{N_{{\mathbf{u}}}|{\mathbf{u}}\in \text{Rowspace}(\hat{F%
}_{S})\},
\end{split}%
\end{equation*}%
up to an overall phase, where
\begin{equation}
\hat{F}_{I}=\left(
\begin{array}{c}
{\mathbf{u}}_{1} \\
\vdots \\
{\mathbf{u}}_{s}%
\end{array}%
\right) ,\ \ \ \hat{F}_{S}=\left(
\begin{array}{c}
{\mathbf{u}}_{s+1} \\
\vdots \\
{\mathbf{u}}_{s+e} \\
{\mathbf{v}}_{s+1} \\
\vdots \\
{\mathbf{v}}_{s+e}%
\end{array}%
\right) .
\end{equation}

We can now specify the error correcting condition in the symplectic
formulation. This EAQECC can correct an error set ${\mathbf{E}}\subset(%
\mathbb{Z}_2)^{2n}$ if for all ${\mathbf{e}}_1, {\mathbf{e}}_2 \in {\mathbf{E%
}}$, either $\hat{F}\odot ({\mathbf{e}}_2-{\mathbf{e}}_1)\neq 0$ or $({%
\mathbf{e}}_2-{\mathbf{e}}_1)\in\text{Rowspace}(\hat{F}_I)$.

\section{Classically Enhanced Quantum Error Correction}

\label{sec3}In this section, we will present a new quantum code that can
transmit both classical and quantum information at the same time.

\subsection{The Stabilizer Formalism}

We define an $[[n,q:c,d;e]]$ entanglement-assisted, classically enhanced
quantum error correction code (EACQ) to be a quantum code which encodes $q$
logical qubits and $c$ classical bits into $n$ physical qubits with the help
of $e$ ebits. Our quantum information is given by the $q$-dimensional state $%
|\phi \rangle \in (\mathcal{H}_{2})^{\otimes q}$, and our classical
information $i\in \left\{ 1,2,\ldots ,2^{c}\right\} $ is represented by a
vector ${\mathbf{x}}_{i}\in (\mathbb{Z}_{2})^{c}$. Here, we keep the
subscript $i$ in ${\mathbf{x}}_{i}$ to remind the reader that ${\mathbf{x}}%
_{i}$ is the binary expression of $i$. Let us denote the $2^{q}$-dimensional
Hilbert space of the original qubits by $\mathcal{H}\equiv (\mathcal{H}%
_{2})^{\otimes q}$, and the subspaces of the $n$-dimensional encoded states
by $\mathcal{C}^{i}$. Our encoding operations $\hat{U}_{enc}^{i}:\mathcal{H}%
\rightarrow \mathcal{C}^{i}$ consist of appending the ancilla states $%
|0\rangle ^{\otimes s}$ and maximally entangled states $|\Phi _{+}\rangle
^{\otimes e}$, where $s+e+q=n$ and $|\Phi _{+}\rangle \equiv \frac{1}{\sqrt{2%
}}\left( \left\vert 00\right\rangle +\left\vert 11\right\rangle \right) ,$
to $|\phi \rangle $ followed by performing the unitary $U_{i}$. Thus, our
encoded states, or "codewords", are defined as
\begin{equation}
|\Psi _{i}\rangle \equiv U_{i}\left( |0\rangle ^{\otimes s}\otimes |\Phi
_{+}\rangle ^{\otimes e}\otimes |\phi \rangle \right) .
\label{encoded state}
\end{equation}%
We require that $\langle \Psi _{i}|\Psi _{j}\rangle =\delta _{ij}$ so that
the classical information is perfectly retrievable.

\begin{theorem}
\label{thm1} We specify an $[[n,q:c,d;e]]$ EACQ by the pair of groups $(%
\mathcal{S}_{Q},\mathcal{S}_{C})$. The quantum stabilizer $\mathcal{S}%
_{Q}=\langle \mathcal{S}_{Q,I},\mathcal{S}_{Q,S}\rangle $ of the code is
generated by $s+2e-c$ elements:
\begin{equation}
\begin{split}
\mathcal{S}_{Q,I}& =\langle \overline{Z}_{c_{1}+1},\overline{Z}%
_{c_{1}+2},\cdots ,\overline{Z}_{s}\rangle , \\
\mathcal{S}_{Q,S}& =\langle \overline{Z}_{s+c_{2}+1},\overline{X}%
_{s+c_{2}+1},\cdots ,\overline{Z}_{s+e},\overline{X}_{s+e}\rangle
\end{split}%
.  \label{S}
\end{equation}%
The classical stabilizer $\mathcal{S}_{C}=\langle \mathcal{S}_{C,I},\mathcal{%
S}_{C,S}\rangle $ of the code is generated by $c$ elements:
\begin{equation}
\begin{split}
\mathcal{S}_{C,I}& =\langle \overline{Z}_{1},\overline{Z}_{2},\cdots ,%
\overline{Z}_{c_{1}}\rangle , \\
\mathcal{S}_{C,S}& =\langle \overline{Z}_{s+1},\cdots ,\overline{Z}%
_{s+c_{2}},\overline{X}_{s+1},\cdots ,\overline{X}_{s+c_{2}}\rangle ,
\end{split}
\label{Sc}
\end{equation}%
where $q+s+e=n$ and $c_{1}+2c_{2}=c$, such that, $\forall g_{j}\in \mathcal{S%
}_{Q}$,
\begin{equation}
g_{j}|\Psi _{i}\rangle =|\Psi _{i}\rangle ,  \label{eq_SQ}
\end{equation}%
and
\begin{equation}
g_{j}^{\prime }|\Psi _{i}\rangle =(-1)^{x_{ij}}|\Psi _{i}\rangle ,
\label{eq_SC}
\end{equation}%
where $g_{j}^{\prime }$ is the $j$-th element of the generator set of $%
\widetilde{\mathcal{S}}_{C}$, which is the Abelian extension of $\mathcal{S}%
_{C}$, and $x_{ij}$ is the $j$-th element of ${\mathbf{x}}_{i}\in (\mathbb{Z}%
_{2})^{c}$.
\end{theorem}

\begin{proof}
We begin with a canonical code that encodes the quantum information $|\phi
\rangle \in (\mathcal{H}_{2})^{\otimes q}$ together with classical
information ${\mathbf{x}}_{i}$ in the following trivial way:
\begin{equation}
\begin{split}
|\phi \rangle \overset{{\mathbf{x}}_{i}}{\longrightarrow }& |\psi
_{i}\rangle =\left( N_{\left( \mathbf{0}|\mathbf{x}_{a}\right) }|0\rangle
^{\otimes c_{1}}\right) |0\rangle ^{\otimes (s-c_{1})} \\
& \left[ \left( N_{\left( \mathbf{x}_{b_{2}}|\mathbf{x}_{b_{1}}\right)
}\otimes I\right) |\Phi _{+}\rangle ^{\otimes c_{2}}\right] |\Phi
_{+}\rangle ^{\otimes e-c_{2}}|\phi \rangle ,
\end{split}
\label{simple transformations}
\end{equation}%
where ${\mathbf{x}}_{a}\in (\mathbb{Z}_{2})^{c_{1}}$ and ${\mathbf{x}}%
_{b_{1}},{\mathbf{x}}_{b_{2}}\in (\mathbb{Z}_{2})^{c_{2}},$ and \textit{I }%
is the $c_{2}\times c_{2}$ identity acting on Bob's qubits. \ Instead of
encoding ${\mathbf{x}}_{i}$ as a whole, we separate ${\mathbf{x}}_{i}$ into $%
{\mathbf{x}}_{a}=x_{i1}\ldots x_{ic_{1}}$, ${\mathbf{x}}_{b_{1}}=x_{i,\left(
c_{1}+1\right) }\ldots x_{i,\left( c_{1}+c_{2}\right) }$, and ${\mathbf{x}}%
_{b_{2}}=$ $x_{i,\left( c_{1}+c_{2}+1\right) }\ldots x_{ic}$ such that $%
c_{1}+2c_{2}=c$, and encode ${\mathbf{x}}_{b_{1}}$ and ${\mathbf{x}}_{b_{2}}$
using $c_{2}$ pairs of maximally entangled states. $\mathbf{x}_{a}$ Clearly,
the set $\{|\psi _{i}\rangle \}$ is stabilized by $\mathcal{S}_{Q}^{\prime
}=\langle \mathcal{S}_{Q,I}^{\prime },\mathcal{S}_{Q,S}^{\prime }\rangle $,
where
\begin{equation}
\begin{split}
\mathcal{S}_{Q,I}^{\prime }& =\langle Z_{c_{1}+1},Z_{c_{1}+2},\cdots
,Z_{s}\rangle , \\
\mathcal{S}_{Q,S}^{\prime }& =\langle Z_{s+c_{2}+1},X_{s+c_{2}+1},\cdots
,Z_{s+e},X_{s+e}\rangle .
\end{split}
\label{dumb}
\end{equation}%
Now let $\mathcal{S}_{C}^{\prime }=\langle \mathcal{S}_{C,I}^{\prime },%
\mathcal{S}_{C,S}^{\prime }\rangle $, where
\begin{equation}
\begin{split}
\mathcal{S}_{C,I}^{\prime }& =\langle Z_{1},\cdots ,Z_{c_{1}}\rangle , \\
\mathcal{S}_{C,S}^{\prime }& =\langle Z_{s+1},\cdots
,Z_{s+c_{2}},X_{s+1},\cdots ,X_{s+c_{2}}\rangle ,
\end{split}%
\end{equation}%
and let $\widetilde{\mathcal{S}}_{C}^{\prime }$ be the Abelian extension of $%
\mathcal{S}_{C}^{\prime }$. Then it is easy to verify that
\begin{equation}
g_{j}^{\prime }|\psi _{i}\rangle =(-1)^{x_{ij}}|\psi _{i}\rangle ,
\end{equation}%
where $g_{j}^{\prime }$ is the $j$-th generator of $\widetilde{\mathcal{S}}%
_{C}^{\prime }$.

Since $(\mathcal{S}_{Q}^{\prime },\mathcal{S}_{C}^{\prime })$ is isomorphic
to $(\mathcal{S}_{Q},\mathcal{S}_{C})$, there exists an unitary $U$ such
that $\mathcal{S}_{Q}=U\mathcal{S}_{Q}^{\prime }U^{\dag }$ and $\mathcal{S}%
_{C}=U\mathcal{S}_{C}^{\prime }U^{\dag }$. The codewords $\{|\Psi
_{i}\rangle \}$ can also be obtained by
\begin{equation}
U|\psi _{i}\rangle =|\Psi _{i}\rangle .  \label{U}
\end{equation}%
It is then easy to verify that (\ref{eq_SQ}) and (\ref{eq_SC}) hold.
\end{proof}

Notice that $\langle\mathcal{S}_{Q},\mathcal{S}_{C}\rangle$ is the
stabilizer of an $[[n,q;e]]$ EAQECC code, and thus it fully specifies one of
the codewords from (\ref{encoded state}), $| \Psi_{\mathbf{0}} \rangle$. For
$c>0,$ the additional codewords are just unitary transformations of $| \Psi_{%
\mathbf{0}} \rangle$. Theorem \ref{thm1} confirms that $\mathcal{S}_C$ and $%
\mathcal{S}_{Q}$ together are sufficient to fully specify the codewords.

Now that we have uniquely defined our code, we will consider the conditions
that make a set of errors correctable, as well as the decoding procedure for
a given set of correctable errors. We will consider here only error sets
which are subsets of $\mathcal{G}^n,$ since it has been shown that the
ability to correct such a discrete error set implies the ability to correct
any linear combination of errors in that set.

\begin{theorem}
\label{cond}A set of errors ${\mathbf{E}}\subset \mathcal{G}^{n}$ is
correctable if for all $E_{m}$,$E_{p}\in {\mathbf{E}}$, $E_{m}^{\dagger
}E_{p}\in \langle \mathcal{S}_{Q,I},\mathcal{S}_{C,I}\rangle \cup (\mathcal{G%
}^{n}-N(\mathcal{S}_{Q}))$, where $N(\mathcal{S})$ is the normalizer of
group $\mathcal{S}$.
\end{theorem}

\begin{proof}
We consider the following different cases.

\begin{enumerate}
\item If $E_{m}^{\dagger }E_{p}\in \mathcal{G}^{n}-N(\mathcal{S}_{Q})$, then
by definition there is at least one element $g_{j}\in \mathcal{S}_{Q}$ such
that
\begin{equation*}
\lbrack E_{m}^{\dagger }E_{p},g_{j}]\neq 0.
\end{equation*}%
Then we are guaranteed that $E_{m}$ and $E_{p}$ have different error
syndromes on the set of codewords $\{|\Psi _{i}\rangle \}$. We can then
perform a recovery operation based on the error syndrome. \ If it is
determined that the error $E_{m}$ occurred, the original codeword may be
recovered by simply performing the unitary $E_{m}$ since $E_{m}\in \mathcal{G%
}^{n}.$

\item If $E_{m}^{\dagger }E_{p}\in N(\mathcal{S}_{Q})$, there are three
cases:

\begin{enumerate}
\item If $E_{m}^{\dagger }E_{p}\in \mathcal{S}_{Q,I}$, then $E_{m}^{\dagger
}E_{p}|\Psi _{i}\rangle =|\Psi _{i}\rangle $. The errors have the same
syndrome, but they also act on the code space the same way. (This is the
case of a \textit{degenerate} code.)

\item If $E_{m}^{\dagger }E_{p}\in \mathcal{S}_{C,I}$, then by (\ref{eq_SC}%
), $E_{m}^{\dagger }E_{p}|\Psi _{i}\rangle =\pm |\Psi _{i}\rangle $. The
errors have the same syndrome, but their effects differ by a possible global
phase without changing the classical information $i$ embedded in the
codeword $|\Psi _{i}\rangle $. Therefore, we can still recover both the
quantum and classical information. (See Theorem \ref{recovery}).

\item For all the rest, the errors act nontrivially on the codewords $%
\{|\Psi _{i}\rangle \}$, but do not have a unique syndrome. If this case
applies to any pair of errors $E_{m},E_{p}\in {\mathbf{E}}$ then the error
set ${\mathbf{E}}$ is uncorrectable.
\end{enumerate}
\end{enumerate}

Combining these cases, we get that whenever $E_{m}^{\dagger }E_{p}\in
\langle \mathcal{S}_{Q,I},\mathcal{S}_{C,I}\rangle \cup (\mathcal{G}^{n}-N(%
\mathcal{S}_{Q}))$ $\forall E_{m},E_{p}\in {\mathbf{E}}$, the codewords $%
\{|\Psi _{i}\rangle \}$ can be recovered up to a possible globe phase.
\end{proof}

\begin{theorem}
\label{recovery} Once error recovery has been performed, the classical index
$i$ may be determined by measuring each of the $g_{k}^{\prime }\in
\widetilde{\mathcal{S}}_{C}$ observables. The original quantum state $|\phi
\rangle $ may be recovered by performing the unitary $U_{i}^{-1}$ and then
discarding the ancillae.
\end{theorem}

\begin{proof}
After we have performed error recovery, the state in our possession will be $%
\pm |\Psi _{i}\rangle $. Measuring the generator set $\{g_{k}^{\prime }\}$
of $\widetilde{\mathcal{S}}_{C}$ will guarantee proper identification of ${%
\mathbf{x}}_{i}$ by (\ref{eq_SC}). Once the classical index has been
identified, we can see from (\ref{encoded state}) that we may recover the
original quantum state $|\phi \rangle $ by performing $U_{i}^{-1}$ and
discarding the states $\pm |0\rangle ^{\otimes s}|\Phi _{+}\rangle ^{\otimes
e}$.
\end{proof}

\subsection{The Symplectic Formalism}

In the following, we will use the symplectic formalism to formulate
this problem and at the same time generalize Theorem 1. The goal
here is to show that actually the EACQs can be completely specified
by some classical parity
check matrix $H$ and quantum parity check matrix $\hat{H}$. Since an $%
[[n,q;e]]$ EAQECC can be defined by a $(s+2e)\times 2n$ quantum parity check
matrix $\hat{F}$ as shown in (\ref{F_hat}), we may specify the quantum
stabilizer $\mathcal{S}_{Q}$ by $\hat{F}$ and a binary matrix $F$:
\begin{widetext}
\begin{equation}\label{F}%
F=\left(
\begin{array}{cc|cc|cc}
\0_{s-c_1\times c_1} & I_{s-c_1\times s-c_1} & \0_{s-c_1\times e-c_2} & \0_{s-c_1\times c_2}& \0_{s-c_1\times e-c_2} & \0_{s-c_1\times c_2}\\
\0_{e-c_2\times s-c_1} & \0_{e-c_2\times c_1} & \0_{e-c_2\times c_2} & I_{e-c_2\times e-c_2} & \0_{e-c_2\times e-c_2} & \0_{e-c_2\times c_2}  \\
\0_{e-c_2\times s-c_1} & \0_{e-c_2\times c_1} & \0_{e-c_2\times e-c_2} & \0_{e-c_2\times c_2} & \0_{e-c_2\times c_2} & I_{e-c_2\times e-c_2}%
\end{array}\right),
\end{equation}
\end{widetext}where $I_{r\times r}$ is the $r\times r$\ identity matrix, and
$\mathbf{0}_{r\times t}$ is the $r\times t$ null matrix. That is,
\begin{equation}
\mathcal{S}_{Q}=\{N_{{\mathbf{v}}}|{\mathbf{v}}\in \text{Rowspace}(\hat{G}%
)\},  \label{sp_SQ}
\end{equation}%
where $\hat{G}=F\hat{F}$.

Now, we may take any full rank, $(s+2e)\times (s+2e)$ matrix $M$ and write
\begin{equation*}
F\hat{F}=(FM)(M^{-1}\hat{F})=H\hat{H},
\end{equation*}%
where $H=FM$ and $\hat{H}=M^{-1}\hat{F}$. Since $M$ is full rank, $\text{%
Rowspace}(\hat{H})=\text{Rowspace}(\hat{F})$, and $\hat{H}$ and $\hat{F}$
specify the same stabilizer. However, $H$ may be any $(s+2e-c)\times (s+2e)$
matrix having linearly independent rows, so $H$ is in fact an arbitrary
classical parity-check matrix!

Although one can always use Theorem \ref{thm1} to specify the code, it may
be somewhat tedious to find the $g_k^{\prime}\in\mathcal{S}_C$ in practice.
Therefore, when formulating a code in the language of parity-check matrices,
it may sometimes be more convenient to use a different set of eigenvalue
equations to take advantage of our \textit{a priori} knowledge of the
properties of the classical parity-check matrix $H$. $H$ specifies a set of $%
2^{c}$ classical codewords $\mathbf{y}_{i}\in(\mathbb{Z}_{2})^{s+2e}$
satisfying
\begin{equation}
H\mathbf{y}_{i}^{T}=\mathbf{0},i=\text{1,2,\ldots2}^{c}\text{.}
\label{kernel of H}
\end{equation}

\begin{theorem}
Assume we are given an $(s+2e)\times 2n$ quantum parity-check matrix $\hat{H}
$ with rows $\mathbf{u}_{l}^{\prime },l=1,2,\ldots ,(s+2e),$ and an $%
(s+2e-c)\times (s+2e)$ classical parity-check matrix $H$ whose kernel is $\{%
\mathbf{y}_{i}\},i=1,2,\ldots ,2^{c}.$ Then we may fully specify the
codewords by the equations, $\forall i,l$,
\begin{equation}
\label{eq_sp} N_{{\mathbf{u}}_{l}^{\prime }}|\Psi _{i}\rangle
=(-1)^{y_{il}}|\Psi _{i}\rangle .
\end{equation}
\end{theorem}

\begin{proof}
Theorem \ref{thm1} can be rewritten as, $\forall i,j$,
\begin{equation*}
N_{{\mathbf{u}}_{j}}|\Psi _{i}\rangle =(-1)^{x_{ij}}|\Psi _{i}\rangle ,
\end{equation*}%
where $\{{\mathbf{x}}_{i}\}$ is the kernel of $F$, and ${\mathbf{u}}_{j}$ is
the $j$-th row of $\hat{F}$. Since $\hat{H}=M^{-1}\hat{F}$, then
\begin{equation}
\begin{split}
N_{{\mathbf{u}}_{l}^{\prime }}|\Psi _{i}\rangle & =\prod_{m=1}^{s+2e}(N_{{%
\mathbf{u}}_{m}})^{(M^{-1})_{lm}}|\Psi _{i}\rangle , \\
& =(-1)^{\sum_{m=1}^{s+2e}(M^{-1})_{lm}x_{im}}|\Psi _{i}\rangle , \\
& =(-1)^{{\mathbf{y}}_{il}}|\Psi _{i}\rangle ,
\end{split}%
\end{equation}%
where ${\mathbf{y}}_{i}=M^{-1}{\mathbf{x}}_{i}$. In order to be valid
codewords, $\{|\Psi _{i}\rangle \}$ must also satisfy $N_{{\mathbf{w}}%
_{j}}|\Psi _{i}\rangle =|\Psi _{i}\rangle $, where ${\mathbf{w}}_{j}$ is the
$j$-th row of $H\hat{H}.$ Then
\begin{eqnarray*}
N_{{\mathbf{w}}_{j}}|\Psi _{i}\rangle &=&\left( \prod\limits_{l=1}^{s+2e}(N_{%
{\mathbf{u}}_{l}^{\prime }})^{H_{jl}}\right) |\Psi _{i}\rangle , \\
&=&(-1)^{\sum_{l=1}^{s+2e}H_{jl}y_{il}}|\Psi _{i}\rangle , \\
&=&(-1)^{0}|\Psi _{i}\rangle =|\Psi _{i}\rangle .
\end{eqnarray*}%
This concludes our proof.
\end{proof}

We have now established a new set of codewords with stabilizer
\begin{equation*}
S_{Q}=\left\{N_{{\mathbf{v}}}|{\mathbf{v}}\in \text{Rowspace}(H\hat{H}%
)\right\},
\end{equation*}
and a new way of specifying the codewords via (\ref{eq_sp}). Theorem \ref%
{cond} was cast in general enough terms that it is applicable to our new
code. So we are now in a position to give the error-correcting conditions
and to explain how to perform error detection and recovery in the language
of the symplectic form as a corollary to Theorem \ref{cond}.

\begin{corollary}
The set of correctable errors ${\mathbf{E}}$ for a code specified by the
quantum parity-check matrix $\hat{H}$ and classical parity-check matrix $H$
are such that for every distinct $N_{{\mathbf{e}}},N_{{\mathbf{e}}^{\prime
}}\in {\mathbf{E}}$, either

\begin{enumerate}
\item ${\mathbf{e}}-{\mathbf{e}}^{\prime }\in \text{Rowspace}(\hat{H}_I)$, or

\item $H\hat{H}\odot\left({\mathbf{e}}-{\mathbf{e}}^{\prime }\right)^{T}\neq%
\mathbf{0}$.
\end{enumerate}

\end{corollary}

\begin{proof}
Since
\begin{equation*}
\langle\mathcal{S}_{Q,I},\mathcal{S}_{C,I}\rangle=\{N_{\mathbf{u}}|{\mathbf{u%
}}\in\text{Rowspace}(\hat{H}_I)\},
\end{equation*}
condition 1 corresponds to
\begin{equation*}
N_{{\mathbf{e}}-{\mathbf{e}}^{\prime }}=N_{{\mathbf{e}}^{\prime
}}^{\dagger}N_{{\mathbf{e}}}\in \langle\mathcal{S}_{Q,I},\mathcal{S}%
_{C,I}\rangle.
\end{equation*}

Let ${\mathbf{v}}_{j}$ denote the $j$-th row of $H\hat{H}$; condition 2 is
equivalent to the statement that for ${\mathbf{e}}$ and ${\mathbf{e}}%
^{\prime }$ there exists a ${\mathbf{v}}_{j}$ such that
\begin{equation*}
\left[ N_{{\mathbf{e}}^{\prime }}^{\dagger }N_{{\mathbf{e}}},N_{{\mathbf{v}}%
_{j}}\right] \neq 0
\end{equation*}%
Therefore, conditions 1 and 2 together are equivalent to $N_{{\mathbf{e}}%
^{\prime }}^{\dagger }N_{{\mathbf{e}}}\in \langle \mathcal{S}_{Q,I},\mathcal{%
S}_{C,I}\rangle \cup (\mathcal{G}^{n}-N(\mathcal{S}_{Q}))$, which are the
error correcting conditions of Theorem \ref{cond}.
\end{proof}

\subsection{Properties of EACQs}

\begin{theorem}
\label{trans1} We can transform any $[[n,q+c,d_1;e]]$ EAQECC code $\mathcal{C%
}_1$ into an $[[n,q:c,d_2;e]]$ EACQ code $\mathcal{C}_2$, and transform any $%
[[n,q:c,d_2;e]]$ EACQ code $\mathcal{C}_2$ into an $[[n,q,d_3;e]]$ EAQECC
code $\mathcal{C}_3$, where $d_1\leq d_2 \leq d_3$.
\end{theorem}

\begin{proof}
The stabilizer group $\mathcal{S}_Q$ of $\mathcal{C}_1$ is of size $2^{s+2e}$%
, where $s+q+c+e=n$. The isotropic subgroup $\mathcal{S}_{Q,I}$ and the
symplectic subgroup $\mathcal{S}_{Q,S}$ of $\mathcal{S}_Q$ is of size $2^s$
and $2^{2e}$, respectively. If we simply add an Abelian group $\mathcal{S}_C$
of size $2^{c}$ such that $\mathcal{S}_C\cap\mathcal{S}_Q=\emptyset$, then $(%
\mathcal{S}_Q,\mathcal{S}_C)$ defines an $[[n,q:c,d_2;e]]$ EACQ code $%
\mathcal{C}_2$ for some $d_2$, which follows from Theorem \ref{thm1}. Let ${%
\mathbf{E}}_1$ be the error set that can be corrected by $\mathcal{C}_1$,
and ${\mathbf{E}}_2$ be the error set that can be corrected by $\mathcal{C}%
_2 $. Clearly, ${\mathbf{E}}_1\subset {\mathbf{E}}_2$ (see table \ref{comp}%
), so $\mathcal{C}_2$ can correct more errors than $\mathcal{C}_1$.
Therefore, $d_2\geq d_1$.

In general, an $[[n,q:c,d_2;e]]$ EACQ code $\mathcal{C}_2$ is defined by $%
\mathcal{S}_Q=\langle\mathcal{S}_{Q,I},\mathcal{S}_{Q,S}\rangle$ and $%
\mathcal{S}_C=\langle\mathcal{S}_{C,I},\mathcal{S}_{C,S}\rangle$, where the
isotropic subgroup $\mathcal{S}_{Q,I}$ and the symplectic subgroup $\mathcal{%
S}_{Q,S}$ of $\mathcal{S}_Q$ is of size $2^{s-c_1}$ and $2^{2(e-c_2)}$,
respectively, and the isotropic subgroup $\mathcal{S}_{C,I}$ and the
symplectic subgroup $\mathcal{S}_{C,S}$ of $\mathcal{S}_C$ is of size $%
2^{c_1}$ and $2^{2c_2}$, respectively. Here the parameters satisfy $s+q+e=n$
and $c_1+2c_2=c$. Now let $\mathcal{S}_{Q,I}^{\prime }=\langle\mathcal{S}%
_{Q,I},\mathcal{S}_{C,I}\rangle$ and $\mathcal{S}_{Q,S}^{\prime }=\langle%
\mathcal{S}_{Q,S},\mathcal{S}_{C,S}\rangle$. Then $\mathcal{S}_Q^{\prime
}=\langle\mathcal{S}_{Q,I}^{\prime },\mathcal{S}_{Q,S}^{\prime }\rangle$
defines an $[[n,q,d_3;e]]$ EAQECC code $\mathcal{C}_3$. Let ${\mathbf{E}}_3$
be the error set that can be corrected by $\mathcal{C}_3$. Let $E\in{\mathbf{%
E}}_2$, then either $E\in\langle \mathcal{S}_{Q,I},\mathcal{S}_{C,I}\rangle$
or $E\not\in N(\mathcal{S}_{Q})$.

\begin{itemize}
\item If $E\in\langle \mathcal{S}_{Q,I},\mathcal{S}_{C,I}\rangle$, then $%
E\in \mathcal{S}_{Q,I}^{\prime }$. Thus, $E\in{\mathbf{E}}_3$.

\item Since $\mathcal{S}_{Q}\subset \mathcal{S}_{Q}^{\prime }$, we have $N(
\mathcal{S}_{Q}^{\prime })\subset N(\mathcal{S}_{Q})$. If $E\not\in N(%
\mathcal{S}_Q)$, then $E\not\in N(\mathcal{S}_Q^{\prime })$. Thus, $E\in{%
\mathbf{E}}_3$.
\end{itemize}

Putting these together we get ${\mathbf{E}}_2\subset {\mathbf{E}}_3$.
Therefore $d_3\geq d_2$.
\end{proof}

It is worth pointing out that the theory of EACQ codes naturally includes
the set of classically enhanced quantum codes that do not require
entanglement as a subclass. These would be codes for which there is no
nontrivial symplectic subgroup for either $\mathcal{S}_{Q}$ or $\mathcal{S}%
_{C}$, so that both of these groups are purely isotropic. In terms of the
parameters describing the code, this is the special case where $e=0$. Our
first example in the next section is exactly such a code.

To conclude this section, we list the different error-correcting criteria of
an EAQECC and an EACQ:
\begin{table}[tph]
\begin{center}
\begin{tabular}{|c|c|}
\hline
EAQECC & EACQ \\ \hline
$E_{m}^{\dagger }E_{p}\not\in N(\langle \mathcal{S}_{Q,I},\mathcal{S}%
_{Q,S}\rangle )$ & $E_{m}^{\dagger }E_{p}\not\in N(\langle \mathcal{S}_{Q,I},%
\mathcal{S}_{Q,S}\rangle )$ \\
$E_{m}^{\dagger }E_{p}\in \mathcal{S}_{Q,I}$ & $E_{m}^{\dagger }E_{p}\in
\langle \mathcal{S}_{Q,I},\mathcal{S}_{C,I}\rangle $ \\ \hline
\end{tabular}
\end{center}
\caption{The error-correcting conditions of EAQECCs and EACQs.}
\label{comp}
\end{table}

\section{Examples}

\label{sec4}

\subsection{$[[9,1:3,3;0]]$ EACQ}

We first give an example of a code that starts from an overly redundant
quantum code, and exploits that redundancy by additionally encoding
classical information. Starting from the 9-qubit Shor code, we modify it to
encode three additional classical bits into the quantum code. The modified
Shor code presented here encodes one qubit and three classical bits into
nine physical qubits, and it is still able to correct an arbitrary error on
a single qubit.

The code is a straightforward combination of the original 9 qubit Shor code,
with parity-check matrix
\begin{equation*}
\hat{H}=\left(
\begin{array}{ccccccccc}
1 & 1 & 0 & 0 & 0 & 0 & 0 & 0 & 0 \\
0 & 1 & 1 & 0 & 0 & 0 & 0 & 0 & 0 \\
0 & 0 & 0 & 1 & 1 & 0 & 0 & 0 & 0 \\
0 & 0 & 0 & 0 & 1 & 1 & 0 & 0 & 0 \\
0 & 0 & 0 & 0 & 0 & 0 & 1 & 1 & 0 \\
0 & 0 & 0 & 0 & 0 & 0 & 0 & 1 & 1 \\
0 & 0 & 0 & 0 & 0 & 0 & 0 & 0 & 0 \\
0 & 0 & 0 & 0 & 0 & 0 & 0 & 0 & 0%
\end{array}
\right\vert \left.
\begin{array}{ccccccccc}
0 & 0 & 0 & 0 & 0 & 0 & 0 & 0 & 0 \\
0 & 0 & 0 & 0 & 0 & 0 & 0 & 0 & 0 \\
0 & 0 & 0 & 0 & 0 & 0 & 0 & 0 & 0 \\
0 & 0 & 0 & 0 & 0 & 0 & 0 & 0 & 0 \\
0 & 0 & 0 & 0 & 0 & 0 & 0 & 0 & 0 \\
0 & 0 & 0 & 0 & 0 & 0 & 0 & 0 & 0 \\
1 & 1 & 1 & 1 & 1 & 1 & 0 & 0 & 0 \\
0 & 0 & 0 & 1 & 1 & 1 & 1 & 1 & 1%
\end{array}
\right) ,
\end{equation*}
and the $[8,3]$ classical code, with parity check matrix
\begin{equation*}
H=%
\begin{pmatrix}
1 & 0 & 1 & 0 & 1 & 0 & 0 & 0 \\
0 & 0 & 0 & 1 & 0 & 1 & 0 & 0 \\
1 & 1 & 1 & 0 & 0 & 0 & 0 & 0 \\
1 & 0 & 0 & 1 & 0 & 0 & 1 & 0 \\
1 & 1 & 1 & 1 & 1 & 1 & 0 & 1%
\end{pmatrix}
.
\end{equation*}
Table \ref{EACQ1} gives the generators of $\mathcal{S}_Q$ and $\mathcal{S}_C$
as in (\ref{eq_SQ}) and (\ref{eq_SC}) for the code.

\begin{table}[htp]
\begin{center}
\begin{tabular}{|c|c|ccccccccc|}
\hline\hline
\multirow{5}{*}{$\cS_Q$} & $g_1$ & Z & Z & I & Z & Z & I & Z & Z & I \\
& $g_2$ & I & I & I & I & Z & Z & I & Z & Z \\
& $g_3$ & Z & I & Z & Z & Z & I & I & I & I \\
& $g_4$ & Y & Y & X & X & Y & Y & I & I & I \\
& $g_5$ & Z & I & Z & Y & X & Y & Y & X & Y \\ \hline
\multirow{3}{*}{$\cS_C$} & $g_1^{\prime }$ & Z & Z & I & I & I & I & I & I &
I \\
& $g_2^{\prime }$ & I & Z & Z & I & I & I & I & I & I \\
& $g_3^{\prime }$ & I & I & I & I & Z & Z & I & I & I \\ \hline\hline
\end{tabular}%
\end{center}
\caption{The resulting $[[9,1:3,3;0]]$ EACQ encodes one qubit and three
classical bits into nine physical qubits.}
\label{EACQ1}
\end{table}

\begin{proposition}
The modified Shor code presented above can correct an arbitrary error on a
single qubit.
\end{proposition}

\begin{proof}
This modified Shor code is degenerate. A single-qubit $Z$ error on any of
the qubits in the same triplet (that is, on any of qubits $1,2,3$, or any of
qubits $4,5,6$, or any of qubits $7,8,9$) result in the same error syndrome,
and can be corrected using the same recovery operation. However, each of the
single-qubit $X$ errors gives a distinct error syndrome, and can therefore
be corrected. The syndromes are obtained by measuring $\{g_{1},\cdots
,g_{5}\}$.
\end{proof}

\subsection{$[[8,1:3,3;1]]$ EACQ code}

The following example comes from modifying the $[[8,1,3;1]]$ EAQECC code
given in \cite{HBD07}. The $[[8,1:3,3;1]]$ EACQ code comes from a
combination of the $[[8,1,3;1]]$ EAQECC with the quantum parity check matrix
\begin{equation*}
\hat{H}=\left(
\begin{array}{cccccccc|cccccccc}
1 & 1 & 0 & 0 & 0 & 0 & 0 & 0 & 0 & 0 & 0 & 0 & 0 & 0 & 0 & 0 \\
1 & 0 & 1 & 0 & 0 & 0 & 0 & 0 & 0 & 0 & 0 & 0 & 0 & 0 & 0 & 0 \\
0 & 0 & 0 & 1 & 1 & 0 & 0 & 0 & 0 & 0 & 0 & 0 & 0 & 0 & 0 & 0 \\
0 & 0 & 0 & 1 & 0 & 1 & 0 & 0 & 0 & 0 & 0 & 0 & 0 & 0 & 0 & 0 \\
0 & 0 & 0 & 0 & 0 & 0 & 1 & 1 & 0 & 0 & 0 & 0 & 0 & 0 & 0 & 0 \\
0 & 0 & 0 & 0 & 0 & 0 & 0 & 0 & 1 & 1 & 1 & 1 & 1 & 1 & 0 & 0 \\
0 & 0 & 0 & 0 & 0 & 0 & 0 & 1 & 0 & 0 & 0 & 0 & 0 & 0 & 0 & 0 \\
0 & 0 & 0 & 0 & 0 & 0 & 0 & 0 & 1 & 1 & 1 & 0 & 0 & 0 & 1 & 1%
\end{array}%
\right) ,
\end{equation*}%
and the $[8,3]$ classical code, with parity check matrix
\begin{equation*}
H=%
\begin{pmatrix}
1 & 0 & 1 & 0 & 1 & 0 & 0 & 0 \\
0 & 1 & 1 & 0 & 0 & 0 & 0 & 0 \\
1 & 0 & 1 & 1 & 0 & 1 & 0 & 0 \\
0 & 0 & 1 & 1 & 0 & 0 & 1 & 0 \\
0 & 1 & 1 & 0 & 1 & 1 & 0 & 1%
\end{pmatrix}%
.
\end{equation*}%
$\hat{H}$ and $H$ together specify $(\mathcal{S}_{Q},\mathcal{S}_{C})$ for
the EACQ given in Table \ref{EACQ2}.
\begin{table}[tph]
\begin{center}
\begin{tabular}{|c|c|cccccccc|}
\hline\hline
\multirow{3}{*}{$\cS_{Q,I}$} & $g_1$ & Z & Z & I & Z & Z & I & Z & Z \\
& $g_2$ & Z & I & Z & Z & Z & I & I & I \\
& $g_3$ & Y & Y & X & X & Y & Y & I & I \\ \hline
\multirow{2}{*}{$\cS_{Q,S}$} & $g_4$ & I & I & I & I & Z & Z & I & Z \\
& $g_5$ & Z & I & Z & Y & Y & X & Y & Y \\ \hline
\multirow{3}{*}{$\cS_C$} & $g_1^{\prime }$ & Z & Z & I & I & I & I & I & I
\\
& $g_2^{\prime }$ & I & Z & Z & I & I & I & I & I \\
& $g_3^{\prime }$ & I & I & I & I & Z & Z & I & I \\ \hline\hline
\end{tabular}%
\end{center}
\caption{The resulting $[[8,1:3,3;1]]$ EACQ encodes one qubit and three
classical bits into eight physical qubits with the help of one ebit.}
\label{EACQ2}
\end{table}
The resulting EACQ encodes one qubit and three classical bits into eight
physical qubits with the help of one ebit. Since the $[[8,1,3;1]]$ code is
derived from the Shor code, this EACQ is clearly related to our first
example.

\subsection{EACQ codes from classical BCH codes}

Here, we will look at the $[[63,21,9;6]]$ EAQECC shown in \cite{HBD07},
which is constructed from a classical binary $[63,39,9]$ BCH code \cite%
{FJM77}. This EAQECC has the interesting property that removing the
symplectic pairs from the quantum parity check matrix will only decrease the
distance from $d=9$ to $d=7$ no matter how many pairs are removed.
Therefore, if we switch all the ebits from $\mathcal{S}_Q$ to $\mathcal{S}_C$%
, we will have a [[63,21:12,7;6]] EACQ. This example shows that it is
possible to encode extra classical information using ebits without degrading
the distance performance too much.

\section{Conclusions}

\label{sec5} In this paper, we have demonstrated yet another extension of
the standard quantum error correction scheme. The new formalism, EACQ, is a
quantum error-correcting code that can transmit both classical and quantum
information simultaneously. We consider this EACQ formalism as both a
generalization and a unification of EAQECCs and classical error correction,
in the following sense:

\begin{itemize}
\item For a purely quantum code ($c=0$), we have $\mathcal{S}_{C}=\emptyset$%
. Then this corresponds to the entanglement-assisted formalism. In this
case, the classical parity check matrix $H$ is chosen to be
\begin{equation*}
H=I_{(n-q)\times (n-q)}
\end{equation*}
such that the quantum parity-check matrix is $\hat{G}=H\hat{H}=\hat{H}$ for
the code.

\item For a purely classical code ($q=0$), we have $\mathcal{S}%
_{Q}=\emptyset $. In this case, the quantum parity check matrix $\hat{H}$ is
chosen to be
\begin{equation*}
\hat{H}=\left(I_{n\times n}|\mathbf{0}_{n\times n}\right)
\end{equation*}
such that the quantum parity-check matrix $\hat{G}=H\hat{H}=(H|\mathbf{0}%
_{n\times n})$ for the code. The classical code can be thought of as encoded
in the $Z$ basis.
\end{itemize}

On the other hand, the EACQ formalism provides further flexibility in the
use of quantum error correcting codes. As shown in the example section, the
EACQ can make use of extra redundancy in quantum codes by encoding
additional classical information. We also note that the passive error
correcting ability of an EACQ is increased at the cost of the quantum code
rate of an EAQECC.

We are currently investigating the relation between EACQs and other
extensions of standard quantum error correction, such as OQECC or ``operator
algebra quantum error correction" (OAQEC) \cite{BKK07}. Recently we are
aware of the work \cite{BKK07QEC}, which also allows correction of hybrid
classical-quantum information based on operator algebra. Given the wider
variety of resources in quantum information theory compared to classical
information theory, we can expect a correspondingly richer set of families
of quantum error-correcting codes.

\begin{acknowledgments}
We wish to acknowledge enlightening discussions with Igor Devetak, and
Cedric Beny. TAB received financial support from NSF Grant No.~CCF-0448658,
and TAB and MHH both received support from NSF Grant No.~ECS-0507270. IK and
MHH received financial support from NSF Grant No.~CCF-0524811 and NSF Grant
No.~CCF-0545845.
\end{acknowledgments}

\bibliographystyle{plain}
\bibliography{ref5}

\end{document}